\documentclass[runningheads]{llncs}
\usepackage[T1]{fontenc}
\usepackage{graphicx}
\graphicspath{ {./images/} }
\usepackage{xcolor}
\usepackage[shortlabels]{enumitem}
\usepackage{amsmath,amsfonts,bm,bbm}

\DeclareMathOperator{\ba}{\bm{a}}
\DeclareMathOperator{\bb}{\bm{b}}

\DeclareMathOperator{\b1}{\mathbbm{1}}
\DeclareMathOperator{\bE}{\mathbb{E}}
\DeclareMathOperator{\bN}{\mathbb{N}}
\DeclareMathOperator{\bR}{\mathbb{R}}
\DeclareMathOperator{\bZ}{\mathbb{Z}}
\DeclareMathOperator{\cA}{\mathcal{A}}
\DeclareMathOperator{\cC}{\mathcal{C}}
\DeclareMathOperator{\cB}{\mathcal{B}}

\DeclareMathOperator{\cG}{\mathcal{G}}
\DeclareMathOperator{\cI}{\mathcal{I}}

\DeclareMathOperator{\cS}{\mathcal{S}}

\DeclareMathOperator{\cU}{\mathcal{U}}

\DeclareMathOperator{\sw}{\mathit{sw}}

\begin{document}
\title{Simultaneous Contests with Equal Sharing Allocation of Prizes: Computational Complexity and Price of Anarchy}
\titlerunning{Simultaneous Contests with Equal Sharing Allocation of Prizes}
%
\author{Edith Elkind \and Abheek Ghosh \and Paul W. Goldberg}
\authorrunning{E. Elkind et al.}
%
\institute{Department of Computer Science, University of Oxford \\
\email{\{edith.elkind,abheek.ghosh,paul.goldberg\}@cs.ox.ac.uk}}
\maketitle              
\begin{abstract}
    We study a general scenario of simultaneous contests that allocate prizes based on equal sharing: each contest awards its prize to all players who satisfy some contest-specific criterion, and the value of this prize to a winner decreases as the number of winners increases. The players produce outputs for a set of activities, and the winning criteria of the contests are based on these outputs. We consider two variations of the model: (i) players have costs for producing outputs; (ii) players do not have costs but have generalized budget constraints. 
    We observe that these games are exact potential games and hence always have a pure-strategy Nash equilibrium. The price of anarchy is $2$ for the budget model, but can be unbounded for the cost model. 
    Our main results are for the computational complexity of these games. We prove that for general versions of the model exactly or approximately computing a best response is NP-hard.
    For natural restricted versions where best response is easy to compute, we show that finding a pure-strategy Nash equilibrium is PLS-complete, and finding a mixed-strategy Nash equilibrium is (PPAD$\cap$PLS)-complete.
    On the other hand, an approximate pure-strategy Nash equilibrium can be found in pseudo-polynomial time.
    These games are a strict but natural subclass of explicit congestion games, but they still have the same equilibrium hardness results.

\keywords{contest theory \and equilibrium analysis \and computational complexity.}

\end{abstract}
\section{Introduction}
Contests are games where players, who are assumed to be strategic, make costly and irreversible investments to win valuable prizes~\cite{vojnovic2015contest}. Typically, the prizes have monetary value, which incentivizes the players to make the costly investments. In some scenarios, the prizes may be associated with reputation and social status. For example, many online forums and websites depend upon user-generated content to provide value to their customers, and award badges, which do not have any monetary value but provide social reputation (e.g., StackOverflow and Quora);
see, e.g., \cite{immorlica2015social}.

In the presence of multiple simultaneous contests, 
each player may explicitly select one or more contests and invest efforts so as to win
the associated prizes. Moreover, sometimes contest participation is implicit:
players engage in various activities, and each contest awards prizes to some of the players based on their performance in a specific subset of activities.  

Consider, for instance, the setting where several social media platforms or news websites compete to attract customers. The potential customers are not homogeneous:
e.g., some may be interested in politics, while others focus on sports or technology.
It is therefore natural to model this setting as a set of simultaneous contests, 
with each individual contest corresponding to a group of customers with similar preferences. The platforms can take actions that make them more attractive to potential customers. Indeed, some of the actions, such as improving the interface, or increasing the update frequency, may impact the platform's performance with respect to several customer groups. That is, we can think of platforms as engaging in several {\em activities}, with their performance in each contest depending on the effort they invest in these activities. Different customer groups may value different mixtures of activities in different ways: e.g., while consumers of financial and sports news care about frequent updates, those who read the gossip column are happy with daily or even weekly updates. Thus, by increasing her investment in an activity, a player may 
improve her performance in several---but not all---contests.

In this work, we study a formal model that can capture scenarios of this type. In our model, there are several players and several simultaneous contests, as well as a set of activities. Each player selects their effort level for each activity (and may incur a cost for doing so, or face budget constraints), and each contest $j$ has its own success function, which specifies combinations of effort for each activity that are sufficient to succeed in $j$. In addition to that, we assume that each contest allocates identical prizes to all agents
that meet its criteria\footnote{The motivation for equal sharing allocation is somewhat similar to that of proportional allocation (e.g., in Tullock contests~\cite{tullock1980efficient}), with equal sharing becoming relevant in situations where the contests use explicit rules (or criteria) to decide on an allocation.}: 
e.g., if, to win contest $j$, it suffices to produce $2$ units along activity $\ell$, then the player who produces $3$ units along $\ell$ and the player who produces $30$ units along $\ell$ receive the same prize from $j$; however,  the value of the prize in contest $j$ may depend on the number of players who 
meet the criteria of $j$. We refer to this setting as {\em multi-activity games}. 

From the perspective of the player, we distinguish between two models: (1) the \textit{cost model}, which has cost functions for producing output, and (2) the \textit{budget model}, which has generalized budget constraints (feasible subsets of the space). The cost functions and the budget constraints may be different across players, capturing the fact that some players may be able to perform better in some activities compared to others.

Our model is very general: we impose very mild and natural constraints on the contests' success functions and prize allocation rules. In particular, we assume that the total value of the prizes allocated to the winners of each contest is a monotonically non-decreasing and concave function of the number of winners;
these assumptions capture the idea that the value of an award decreases as it gets awarded to a larger number of players (see, e.g., \cite{immorlica2015social}).
A specific instantiation of this model is a contest that has a \textit{fixed total prize} of $V$, and every winner gets a share of $V/k$ if there are $k$ winners.

Besides our social media platform example, several other situations can be modeled using contests with equal sharing allocation of prizes, where players compete by engaging in activities valued by multiple contests. For instance,  funding agencies are generally not able to \textit{perfectly discriminate} among the applicants to select the most deserving ones. They might score candidates based on attributes such as strength of proposal, publication history, and management experience, and allocate their budget to one of the eligible candidates or distribute it among them. Because of this uniform distribution of the funding, from the perspective of an expected utility maximizing applicant, this situation can be approximated using a contest with an equal sharing allocation with a fixed total prize.

\subsection{Our Results}
We observe that, under mild assumptions, multi-activity games are \textit{exact potential games}~\cite{monderer1996potential} (see Definition~\ref{def:potentialGames}). This guarantees the existence of a pure-strategy Nash equilibrium (PNE). 
Moreover, an approximate PNE can be computed in pseudo-polynomial time. In fact, a sequence of $\epsilon$-better/best\footnote{An $\epsilon$-better response move increases the utility of a player by at least $\epsilon$.} response moves converges to an $\epsilon$-PNE, with the number of steps proportional to $1/\epsilon$ (and pseudo-polynomial in other parameters).
For the budget model, for the natural definition of social welfare in these games, we observe that the price of anarchy (PoA) is at most $2$ and the price of stability (PoS) can be close to $2$, so both these bounds are tight.\footnote{The price of anarchy (stability) is the ratio between the social welfare of the optimal solution and the social welfare of an equilibrium solution, in the worst case over instances of the problem, and in the worst (best) case over corresponding equilibrium solutions.} 
However, for the cost model, for meaningful definitions of social welfare, the PoS can be infinite. 

%


We then study the computation complexity of finding an equilibrium in these games, which is the main focus of our paper. This portion of the paper concentrates on a specific instantiation of the model, where the contests' criteria and the players' budget constraints are all linear; let us call this model the \textit{linear model}. (Here, we discuss the results for the budget model. With similar assumptions, similar results hold for the cost model as well.)
We show that, for the linear model, it is NP-hard for a player to best respond to the strategies of other players. This hardness result holds even for a game with only one player; in other words, the hardness is due to the optimization problem that a player faces while playing the game. We also prove that there exists no polynomial-time approximation scheme unless P$=$NP. (Here, NP-hardness for best-response directly implies NP-hardness for equilibrium computation.) On the positive side, we obtain fixed-parameter tractability results: we show that best response is polynomial-time computable if either the number of contests or the number of activities is a constant.

The NP-hardness result for best-response motivates us to further restrict our model: 
we assume that a player can produce output only along a small (polynomial) number of portfolios of activities. Mathematically, a portfolio corresponds to a direction in the activity space along which a player can produce output. This restricted model captures salient features of the original model: e.g., it maintains the property that the contests can have overlapping criteria.

With a simple transformation of the activity space, the portfolio model can be converted to an equivalent model where a player produces output along a single activity only, i.e., only along the axes, where an activity in the new model corresponds to a portfolio in the old model. In our discussion, we call this new model the \textit{single-activity model} to differentiate it from the original model, which we call the \textit{multi-activity model}.

The positive results for the multi-activity model automatically carry over to the single-activity model. Additionally, it is computationally easy for the players to best respond in the single-activity model. However, we get a different hardness result. Even for the linear model, it is PLS-complete to compute a PNE and (PPAD$\cap$PLS)-complete to compute a mixed-strategy Nash equilibrium (MNE).
These hardness results, particularly the (PPAD$\cap$PLS)-hardness result, are interesting because single-activity games form a strict,  structured and well-motivated subclass of explicit congestion games (a contest awards a $1/k$ fraction of a fixed prize to each winner if there are $k$ winners, but a congestion game can have a cost that is an arbitrary function of the number of winners), yet finding an MNE in these games has the same computational complexity as finding an MNE of explicit congestion games~\cite{babichenko2021settling} (and finding a fixed-point of gradient descent~\cite{fearnley2021complexity}).
We also prove some fixed-parameter tractability results with respect to the number of players and the number of contests.


The rest of this paper is organized as follows. 
After summarizing the related work (Section~\ref{sec:related}), in Section~\ref{sec:modelpnepoa} we introduce the general multi-activity model. We also prove the existence of PNE, the pseudo-polynomial convergence of $\epsilon$-best-response dynamics to $\epsilon$-PNE, and present our PoA results. In Section~\ref{sec:multiActivity} we establish the hardness of best-response in linear multi-activity models. Section~\ref{sec:singleActivity} focuses on the single-activity model and presents our results on PLS-completeness and (PPAD$\cap$PLS)-completeness. 
Some proofs are omitted due to space constraints.

\subsection{Related Work}\label{sec:related}
The model of simultaneous contests with equal sharing allocation of prizes has been studied before in the literature~\cite{may2014filter,vojnovic2015contest}. At a high level, our contribution is to (i) generalize the model and extend the positive results and (ii) study the complexity of computing equilibria.

The linear budget model with a fixed total prize has previously been studied by May et al.~\cite{may2014filter}, to model situations such as the social media platform example discussed earlier. Their theoretical results are similar to our positive results: 
(i) they prove existence of PNE by showing that the game is an exact potential game; (ii) they establish a PoA bound of $2$. For (i), existence of PNE, we give a simpler proof by explicitly constructing the potential function that lifts these results to our general model. Our proof also makes it transparent that the PNE exists because of the equal sharing property of the contests (the \textit{congestion} property), and the other restrictions of the model of May et al.~\cite{may2014filter}---the linear budget constraints, the linear criteria of contests, and the fixed total prize---are not necessary for the result. Moreover, the proof clarifies that using the budget model is also non-essential, as the result holds for the cost model as well. For (ii), the PoA bound, May et al.~\cite{may2014filter} prove the result from first principles. In contrast, we use the result of Vetta~\cite{vetta2002nash} for submodular social welfare functions to derive the same result in our---more general---setting.
In summary, we extend the positive results of \cite{may2014filter} to a general model (and we also study computational complexity, which was not considered by May et al.~\cite{may2014filter}).
May et al.~\cite{may2014filter} perform an empirical study and show that the real-life behavior of social media curators resembles the predictions of the model. 
  Bilo et al.~\cite{bilo2019project} consider a model similar to the single-activity model of our paper and show inclusion in the class PLS (but no hardness results).


Models of simultaneous contests that do not have activities (i.e., the players directly produce outputs for each contest, or, equivalently, there is a one-to-one mapping between the activities and the contests) have been extensively studied. Cost-based models where the prizes are awarded based on the players' ranks have been considered by a number of authors \cite{amann2008parallel,archak2009optimal,bapna2010vertically,dipalantino2009crowdsourcing}, including empirical work \cite{archak2009optimal,dipalantino2009crowdsourcing,lakhani2010topcoder,yang2008crowdsourcing}. Colonel Blotto games, where the players have budget constraints and the prize is awarded to the highest-output player for each contest, were proposed by Borel~\cite{borel1921analyse}, and have received a significant amount of attention in the literature (e.g., \cite{borel1938application,tukey1949problem,gross1950continuous,blackett1954some,blackett1958pure,bellman1969colonel,shubik1981systems,roberson2006colonel,hart2008discrete,adamo2009blotto,roberson2006colonel,hortala2012pure,washburn2013or}). Simultaneous contests with proportional allocation have been studied by, e.g., \cite{feldman2005price,zhang2005efficiency,palvolgyi2010strategic}.

Two very recent related papers are by Birmpas et al.~\cite{birmpas2022parallel} and Elkind et al.~\cite{elkind2021contest}. Both these papers do not have activities, i.e., the players produce output directly for the contests, which makes their models a bit different (simpler) than ours, but they add complexity along other dimensions. Therefore, their results are not directly comparable to ours. Birmpas et al.~\cite{birmpas2022parallel} have both budgets and costs in the same model, and they give a constant factor PoA bound by augmenting players' budgets when computing the equilibrium welfare (but not when computing the optimal welfare). 
Elkind et al.~\cite{elkind2021contest} consider a model with only one contest and in the case of incomplete information. Their focus is on mechanism design, and for one of the objectives studied in the paper, they prove that the optimal contest distributes its prize equally to all players who produce output above some threshold, similar to the contests in our paper.

The complexity class PLS (Polynomial Local Search) and the concepts of PLS-hardness and PLS-completeness were introduced by Johnson et al.~\cite{johnson1988easy}. PLS consists of discrete local optimization problems whose solutions are easy to verify (the cost of a given solution can be computed in polynomial time and its local neighborhood can be searched in polynomial time). Similar to NP-hard problems, PLS-hard problems are believed to be not solvable in polynomial time. Several natural problems, such as finding a locally optimal solution of {\sc Max-Cut}, were shown to be PLS-complete by Schaffer~\cite{schaffer1991simple}. The problem of finding a PNE in explicit congestion games (which always have a PNE) is also PLS-complete~\cite{fabrikant2004complexity}, from which it follows that better or best response dynamics take an exponential time to converge in the worst case~\cite{johnson1988easy}.

The class PPAD (Polynomial Parity Arguments on Directed graph) was introduced by Papadimitriou~\cite{papadimitriou1994complexity}. Like PLS problems, PPAD problems always have solutions. For PPAD, the existence of a solution is based on a parity argument: In a directed graph where each vertex has at most one predecessor and one successor, 
if there exists a source vertex (i.e., a vertex with no predecessor), then
there exists some other degree-1 vertex.
One of the most well-known results in algorithmic game theory
is that the problem of finding a mixed-strategy Nash equilibrium (MNE) is PPAD-complete~\cite{daskalakis2009complexity,chen2009settling}. 

Recent work has determined the complexity of computing an MNE of an explicit congestion game. The class PPAD$\cap$PLS represents problems that can be solved both by an algorithm that solves PPAD problems and by an algorithm that solves PLS problems. 
Finding an MNE of an explicit congestion game is in PPAD$\cap$PLS: indeed,  this problem can be solved either by finding a PNE (which is also an MNE) using an algorithm for PLS problems or by computing an MNE using an algorithm for PPAD problems.
Recently, Fearnley et al.~\cite{fearnley2021complexity} proved that finding a fixed-point of a smooth $2$-dimensional function $f : [0,1]^2 \rightarrow \bR$ using gradient descent is complete for the class PPAD$\cap$PLS; based on this result, Babichenko and Rubinstein~\cite{babichenko2021settling} proved that finding an MNE of an explicit congestion game is also (PPAD$\cap$PLS)-complete.
PLS-complete, PPAD-complete, and (PPAD$\cap$PLS)-complete problems are considered to be hard problems with no known polynomial-time algorithms.
\section{General Model, Pure Nash Equilibrium, and Price of Anarchy}\label{sec:modelpnepoa}
In this section, we formally define the general multi-activity model, prove the existence of pure Nash equilibria, and show that the price of anarchy (PoA) is $2$ for the budget variant of the model and $+\infty$ for the cost variant.

We consider a set of $n$ players, $N = [n]$,\footnote{Let $[\ell] = \{1, 2, \ldots, \ell\}$ for any positive integer $\ell \in \bZ_{>0}$.} who simultaneously produce output along $k$ activities, $K = [k]$. There are $m$ contests, $M = [m]$, which award prizes to the players based on their outputs. The contests may have different prizes and may value the activities differently.

We study two models: in one, the players have output production costs; and in the other, the players have generalized output budgets.
Player $i \in N$ chooses an output vector $\bb_i = (b_{i,\ell})_{\ell \in K} \in \bR_{\ge 0}^k$.
\begin{itemize}
    \item \textbf{Cost.} Player $i$ incurs a cost of $c_i(\bb_i)$ for producing $\bb_i$, where $c_i : \bR_{\ge 0}^k \rightarrow \bR_{\ge 0}$ is a non-decreasing cost function with $c_i(\bm{0}) = 0$ (normalized). 
    \item \textbf{Budget.} The player does not incur any cost for the output vector $\bb_i$, but $\bb_i$ is restricted to be in a set $\cB_i \subseteq \bR_{\ge 0}^k$. We assume that $\bm{0} \in \cB_i$, i.e., players are always allowed to not participate.\footnote{The budget model has an alternative interpretation---each player selects a subset of activities among a feasible set of subset of activities for that player---as discussed in Section~\ref{sec:poa}.}
\end{itemize}
Technically, the budget model is a special case of the cost model, and all the positive results for the cost model automatically carry over to the budget model. We make the distinction because when we study the restricted linear models in Section~\ref{sec:multiActivity}, the linear cost model and the linear budget model will have different formulations.

Let $\bb = (\bb_i)_{i \in N} = (b_{i,\ell})_{i \in N, \ell \in K}$. 
Each contest $j\in M$ is associated with a pair of functions 
$f_j: \bR_{\ge 0}^k \rightarrow \bR_{\ge 0}$ 
and
$v_j: \bN \rightarrow \bR_{\ge 0}$.
The function $f_j$, which is an increasing function such that $f_j(\bm{0}) = 0$, determines the set of winners of contest $j$: we say that player $i$ {\em wins} contest $j$ if $f_j(\bb_i) \ge 1$ and set $N_j(\bb) = \{ i\in N \mid f_j(\bb_i) \ge 1 \}$.
 Let $n_j(\bb) = |N_j(\bb)|$. 
 The function $v_j$ determines how the prizes are allocated: each player 
 in $N_j(\bb)$ receives a prize of $v_j(n_j(\bb))$. 
 The total prize allocated by contest $j$ is then $n_j(\bb)\cdot v_j(n_j(\bb))$. We make the following assumptions about the function $v_j$, which are necessary for our price of anarchy bounds.
 \begin{enumerate}
     \item $v_j(\ell)$ is a non-increasing function of $\ell$: 
     $v_j(\ell) \ge v_j(\ell + 1)$. 
     \item $\ell \cdot v_j(\ell)$ is a non-decreasing function of $\ell$: 
     $\ell \cdot v_j(\ell) \le (\ell + 1) \cdot v_j(\ell + 1)$. 
     \item $\ell \cdot v_j(\ell)$ is a weakly concave function of $\ell$, i.e., $(\ell + 1) \cdot v_j(\ell + 1) - \ell \cdot v_j(\ell)$ is a non-increasing function of $\ell$:
     $(\ell+2) \cdot v_j(\ell+2) - 2 \cdot (\ell+1) \cdot v_j((\ell+1) + \ell \cdot v_j(\ell) \le 0$.
     This condition says that the rate of increase in the total prize allocated by a contest $j$ weakly decreases as the number of winners increases.
 \end{enumerate}
Some examples of functions $v_j$ that satisfy these conditions are:
\begin{itemize}
    \item $x \cdot v_j(x) = x$ or $v_j(x) = 1$. Here, the total value of the prize scales linearly with the number of winners, i.e., the prize awarded to a winner does not change as the number of winners increases.
    \item $x \cdot v_j(x) = 1$ or $v_j(x) = 1/x$. Here, the total value of the prize remains constant, i.e., the prize awarded to a winner decreases and is equal to the inverse of the number of winners.
    \item $x \cdot v_j(x) = \sqrt{x}$ or $v_j(x) = 1/\sqrt{x}$. This sits between the previous two examples. Here, the total value of the prize increases, but the prize awarded to a winner decreases as the number of winners increases.
\end{itemize}

The utility of a player $i$ in the cost and the budget model is, respectively,
\[
    u^C_i(\bb) = \sum_{j \in M} v_j(n_j(\bb)) \cdot \b1_{\{i \in N_j(\bb)\}} - c_i(\bb_i),\quad u^B_i(\bb) = \sum_{j \in M} v_j(n_j(\bb))\cdot \b1_{\{i \in N_j(\bb)\}}
\]
where $\b1_{\{\ldots\}}$ is the indicator function; we omit the superscripts $B$ and $C$ if they are clear from the context.

We will be interested in stable outcomes of multi-activity games, 
as captured by their Nash equilibria.

\begin{definition}[Pure-Strategy Nash Equilibrium (PNE)]
A pure strategy profile $\bb = (\bb_i, \bb_{-i})$ is a {\em pure Nash equilibrium (PNE)} if for every $i \in N$ and every action $\bb_i'$ of player $i$ we have $u_i(\bb) \ge u_i(\bb_i', \bb_{-i})$.
\end{definition}
\begin{definition}[Mixed-Strategy Nash Equilibrium (MNE)]
A mixed strategy profile $\mu = \times_{\ell \in [n]} \mu_{\ell}$ is a {\em mixed Nash equilibrium (MNE)} if for every $i \in N$ and every distribution over actions $\mu_i'$ of player $i$ we have\\
$\bE_{\bb \sim \mu}[u_i(\bb)] \ge \bE_{\bb_i' \sim \mu_i', \bb_{-i} \sim \mu_{-i}}[u_i(\bb_i', \bb_{-i})]$, where $\mu_{-i} = \times_{\ell \neq i} \mu_{\ell}$.
\end{definition}

We now give definitions of the price of anarchy and the price of stability, which will be used to study the efficiency of equilibria in our models.
\begin{definition}[Price of Anarchy (PoA) and Price of Stability (PoS)]\label{def:poaPos}
Let $\cG$ denote a class of games and let $\cI \in \cG$ denote a particular instance of the game.
For a given instance of a game $\cI$, let $Action(\cI)$ denote the set of action profiles and $Eq(\cI)$ denote the set of all MNE. Let $sw(\bb)$ denote the social welfare for an action profile $\bb \in Action(\cI)$. The {\em price of anarchy} and the {\em price of stability} for class $\cG$ are defined as, respectively,
\[
    PoA = \max_{\cI \in \cG} \frac{\max_{\bb \in Action(\cI)} sw(\bb) }{ \min_{\mu \in Eq(\cI)} \bE_{\bb \sim \mu}[ sw(\bb) ]};\quad PoS = \max_{\cI \in \cG} \frac{\max_{\bb \in Action(\cI)} sw(\bb) }{ \max_{\mu \in Eq(\cI)} \bE_{\bb \sim \mu}[ sw(\bb) ]}. 
\]
I.e., the PoA (PoS) is the ratio between the optimal social welfare and the equilibrium social welfare in the worst (best) case over possible equilibria and in the worst case over instances of the game.
\end{definition}

\subsection{Existence of Pure-Strategy Nash Equilibrium}
We start by showing that multi-activity games are exact potential games.
We then use the classic result of 
Monderer and Shapley~\cite{monderer1996potential} to conclude that multi-activity games
always have pure Nash equilibria. 
\begin{definition}[Exact Potential Games]\label{def:potentialGames}\cite{monderer1996potential}
A normal form game is an {\em exact potential game} if there exists a potential function $\phi$ such that for any player $i$ with utility function $u_i$,
any two strategies $\bb_i$ and $\bb_i'$ of player $i$, and any strategy profile $\bb_{-i}$ of the other players it holds that
\[ u_i(\bb_i', \bb_{-i}) - u_i(\bb_i, \bb_{-i}) = \phi(\bb_i', \bb_{-i}) - \phi(\bb_i, \bb_{-i}). \]
\end{definition}
\begin{theorem}\cite{monderer1996potential}
Exact potential games always have a pure-strategy Nash equilibrium.
Indeed, every pure strategy profile that maximizes $\phi$ 
is a PNE.
\end{theorem}

We define the potential functions for the multi-activity budget and cost games as, respectively,
\begin{equation}\label{eq:potFun}
    \phi^B(\bb) = \sum_{j \in M} \sum_{\ell \in [n_j(\bb)]} v_j(\ell),\qquad
    \phi^C(\bb) = \sum_{j \in M} \sum_{\ell \in [n_j(\bb)]} v_j(\ell) - \sum_{i \in N} c_i(\bb_i);
\end{equation}
we omit the superscripts $B$ and $C$ if they are clear from the context. We use these potential functions to prove the existence of PNE in multi-activity games. In Section~\ref{sec:singleActivity}, we shall also use them to study the complexity of computing equilibria in these games. 
\begin{theorem}\label{thm:multPNE}
A multi-activity budget/cost game is an exact potential game, and hence has a pure-strategy Nash equilibrium.
\end{theorem}
We note that the crucial property required for the proof of Theorem~\ref{thm:multPNE} is the equal sharing property; some of the other assumptions made in our model, e.g., that the cost functions $c_i(b_i)$ are non-decreasing, $\ell \cdot v_j(\ell)$ are non-decreasing and weakly concave, etc., are not essential for the proof. We also note that for a restricted version of the model the same result was proved by \cite{may2014filter,vojnovic2015contest}, but we believe that our proof using the potential function is simpler.

\subsection{Approximate Pure-Strategy Nash Equilibrium using Better-Response}
The characterization in Theorem~\ref{thm:multPNE} 
is very useful for finding approximate equilibria 
of multi-activity games. 
\begin{definition}[$\epsilon$-Pure-Strategy Nash Equilibrium]
A pure strategy profile $\bb = (\bb_i, \bb_{-i})$ is an {\em $\epsilon$-PNE} if for every $i \in N$ and every action $\bb_i'$ of player $i$ we have $u_i(\bb_i', \bb_{-i}) \le u_i(\bb) - \epsilon$.
\end{definition}
\begin{definition}[$\epsilon$-Better-Response]
For a pure strategy profile $\bb = (\bb_i, \bb_{-i})$, a player $i \in N$, and an action $\bb_i'$ of player $i$, the move from $\bb_i$ to $\bb_i'$ is an {\em $\epsilon$-better-response move} if $u_i(\bb_i', \bb_{-i}) > u_i(\bb) + \epsilon$.
\end{definition}
From the definitions above, it is immediate that a pure strategy profile is an $\epsilon$-PNE  if and only if it does not admit any $\epsilon$-better-response moves.  Now, as multi-activity games are exact potential games,
each $\epsilon$-better-response increases the potential 
by at least $\epsilon$. As the potential function is bounded from above, a sequence of $\epsilon$-better-response moves necessarily terminates, 
and the resulting profile is an $\epsilon$-PNE.
\begin{corollary}\label{thm:pseudoPolyBR}
In multi-activity games, any sequence of $\epsilon$-better-response moves arrives to an $\epsilon$-PNE in at most $n \cdot m \cdot (\max_{j} v_j(1)) / \epsilon$ steps.
\end{corollary}

Corollary~\ref{thm:pseudoPolyBR} provides a pseudo-polynomial bound in the number of steps required for convergence to an approximate PNE, which can be meaningful in situations where the prizes (i.e., the values $v_j(1)$) are not very large. It also implies that an $\epsilon$-PNE can be computed in pseudo-polynomial time if we can compute $\epsilon$-better-responses efficiently.

\subsection{Social Welfare, Price of Anarchy, and Price of Stability}\label{sec:poa}
To start, let us briefly discuss what would be a natural definition(s) of social welfare for multi-activity games.

Consider first the budget model. From the perspective of the players, 
social welfare naturally corresponds to the total prize allocated. 
This definition is also natural from the perspective of the contests.
Indeed, consider the motivating example involving the social media curators (players) and subscribers (contests). The welfare of the curators corresponds to the total subscribers' attention they receive, which is equal to the total prize allocated. On the other hand, the welfare of the subscribers corresponds to the compatible curators who serve them, which again corresponds to the total prize allocated. For this definition of social welfare, we prove an upper bound of $2$ on the price of anarchy (PoA) and a lower bound of $2-o(1)$ on the price of stability (PoS), so both of these bounds are tight. 

For the cost model, the formulation of social welfare can be similarly motivated, with or without subtracting the cost. If one assumes that the cost to the players is a sunk cost, then it is reasonable to subtract it from the social welfare. On the other hand, if one assumes that the cost gets transferred to the contest organizers (or the society), then it should not be subtracted. In any case, for both of these definitions, the PoS (and therefore, PoA) turns out to be infinite.

\subsubsection{Budget Model}
For the budget model, the social welfare is equal to the total prize that gets allocated: 
\begin{equation}\label{eq:SocialWelfareBud}
    \sw(\bb) = \sum_{i \in N} u_i(\bb) = \sum_{i \in N} \sum_{j \in M} v_j(n_j(\bb)) \cdot\b1_{\{i \in N_j(\bb)\}} = \sum_{j \in M} v_j(n_j(\bb))\cdot n_j(\bb) .
\end{equation}
To prove the upper bound of 2 on PoA, we shall use the result of Vetta~\cite{vetta2002nash} for \textit{submodular} social welfare functions. This result states that the PoA can be upper-bounded by $2$ if the following conditions are satisfied:
\begin{enumerate}
    \item the utility of the players and the social welfare are measured in the same units,
    \item the total utility of the players is at most the social welfare,
    \item the social welfare function is non-decreasing and submodular, and
    \item the private utility of a player is at least as much as the \textit{Vickrey utility} (see Definition~\ref{def:vickrey}).
\end{enumerate}
By definition, our model satisfies the first two requirements on this list. To show that it also satisfies the last two requirements, we reinterpret it using a different notation.

In the budget game, an action by a player $i$ effectively corresponds to a subset of the contests that $i$ wins. The output vector produced by the player to win a given set of contests is not important as long as the set of contests that she wins remains the same.\footnote{Not important for PoA analysis, but very important for computational complexity: the budget constraint for a player $i$ boils down to selecting a subset of contests in some feasible set of subsets, say $\cA_i \in 2^M$. The budget constraint is a concise way of representing this $\cA_i$, and $\cA_i$ may be of size exponential in the representation.} Hence, we will represent the action taken by player $i$ as a subset of elements from a set of size $m$, as follows:
Let $\cG$ denote a set of $n \cdot m$ elements partitioned into $n$ size-$m$ sets $\cG^1, \cG^2, \ldots, \cG^n$, where $\cG^i = \{g^i_1, g^i_2, \ldots, g^i_m\}$. We represent the set of feasible actions for player $i$ by $\cA_i \subseteq 2^{\cG^i}$: an action $A_i \in\cA_i$
contains $g^i_j$ if and only if player $i$ satisfies the criteria for contest $j \in M$.
Note that $\emptyset \in \cA_i$, because we assume that each player is allowed to not participate and produce a $\bm{0}$ output (which does not win her any contests). An action profile $A = (A_i)_{i \in N} \in \times_{i \in N} \cA_i $ can be equivalently written as $\cup_{i \in N} A_i \in \cup_{i \in N} \cA_i$ because the sets $\cA_i$ are disjoint as per our notation.

As before, let $N_j(A) = \{ i \mid g^i_j \in A \}$ denote the players who win contest $j$ under action profile $A$, and let $n_j(A) = |N_j(A)|$.
The utility functions and the social welfare function can be rewritten using the new set notation as follows:
\begin{align*}
    u_i(A) = \sum_{j \in M} v_j( n_j(A) ) \cdot \b1_{\{g^i_j \in A\}};\qquad \sw(A) =  \sum_{j \in M} v_j( n_j(A) ) \cdot n_j(A).
\end{align*}

Let us formally define submodular set functions.
\begin{definition}[Submodular Functions]
A function $f: 2^{\Omega} \rightarrow \bR$ is {\em submodular} if 
for every pair of subsets $S, T \subseteq \Omega$ such that $S \subseteq T$
and every $x \in \Omega \setminus T$ we have 
\[
    f(S \cup \{x\}) - f(S) \ge f(T \cup \{x\}) - f(T).
\]
\end{definition}

The next two lemmas prove that our model satisfies the conditions required to use the result of Vetta~\cite{vetta2002nash}.
\begin{lemma}\label{thm:PoA:1}
For the budget model, if the total prize allocated by a contest is a weakly concave function of the number of winners of the contest, then the social welfare function is submodular.
\end{lemma}

\begin{definition}[Vickrey Utility]\label{def:vickrey}
The {\em Vickrey utility} of player $i$ 
at an action profile $A$ is the loss incurred by other players due to $i$'s participation, i.e., 
\begin{equation*}
    u_i^\text{Vickrey}(A) = \sum_{q \neq i} u_{q}(\emptyset, A_{-i}) - \sum_{q \neq i} u_{q}(A_i, A_{-i}),
\end{equation*}
where $\emptyset$ is the action to not participate (or, equivalently, produce $0$ output and not win any contests).
\end{definition}

\begin{lemma}\label{thm:PoA:2}
For the budget model, if the total prize allocated by a contest is a non-decreasing function of the number of winners of the contest, then 
for every profile $A$ and every player $i$, the utility of $i$ is at least as large as her Vickrey utility, i.e., $u_i(A_i, A_{-i}) \ge u_i^\text{Vickrey}(A)$.
\end{lemma}

\begin{theorem}\label{thm:swBudMul}
The social welfare in any mixed-strategy Nash equilibrium of a multi-activity budget game is at least $1/2$ of the optimum social welfare.
\end{theorem}
The following result complements the upper bound of $2$ for PoA by giving a lower bound of $2$ for PoS, which is based upon an example in \cite{vojnovic2015contest} and ensuring that the equilibrium is unique.
\begin{theorem}\label{thm:swBudMulNeg}
There are instances of multi-activity budget games where the social welfare in every mixed-strategy Nash equilibrium approaches $1/2$ of the optimum social welfare as the number of players grows.
\end{theorem}

\subsubsection{Cost Model}
As discussed before, we consider two definitions of social welfare for the cost model. The first definition does not subtract the costs from social welfare and corresponds to the definition given in \eqref{eq:SocialWelfareBud}. The second definition subtracts the costs from the social welfare and is equal to
\begin{equation}\label{eq:SocialWelfareCost}
    \overline{\sw}(\bb) = \sum_{i \in N} (u_i(\bb) - c_i(\bb)) 
    = \sum_{j \in M} v_j(n_j(\bb)) \cdot n_j(\bb) - \sum_{i \in N} c_i(\bb).
\end{equation}
Next, we prove that the PoS can be unbounded for both definitions of social welfare in the cost model.
\begin{theorem}\label{thm:swCost}
There are instances of multi-activity cost games where the social welfare in every mixed-strategy Nash equilibrium can be arbitrarily low compared to the optimum social welfare. 
This holds even if there are at most two players, two activities, and two contests.
\end{theorem}

\section{Multi-Activity Games: Hardness of Best-Response}\label{sec:multiActivity}
In this section, we focus on a restricted model: each contest uses a linear criterion, and the budget constraint or the cost function of each player are also linear. We call this model the {\em linear multi-activity model}. Formally, for an output profile $\bb = (\bb_i)_{i \in N} = (b_{i,\ell})_{i \in N, \ell \in K}$, the winners of contest $j$ are $N_j(\bb) = \{ i \mid \sum_{\ell \in K} w_{j,\ell} b_{i,\ell} \ge 1 \}$,
where $w_{j,\ell} \in \bR_{\ge 0}$ is a non-negative weight that contest $j$ has for activity $\ell$. 
Similarly, the linear budget constraint of a player $i$ is of the form $\sum_{\ell \in K} \beta_{i, \ell} b_{i, \ell} \le 1$, where $\beta_{i,\ell} \in \bR_{> 0}$. Likewise, a linear cost function for player $i$ is of the form $\sum_{\ell \in K} c_{i,\ell} b_{i, \ell}$, where $c_{i,\ell} \in \bR_{> 0}$.

We also impose another constraint: we assume that for each contest its total prize is fixed. That is, each contest $j \in M$ is associated with a total prize $V_j$, and if there are $\ell$ winners, each winner gets a prize $V_j / \ell$.

We study the computational complexity of best-response in this linear multi-activity model. Observe that it suffices to consider this problem for $n=1$. Indeed, consider a player $i\in N$. If there are $\ell$ winners other than $i$ for a given contest $j$, then $i$ gets a prize of $v_j(\ell+1) = V_j/(\ell+1)$ from this contest if she satisfies this contest's criteria, and $0$ otherwise. By scaling the values of all contests appropriately, we reduce $i$'s optimization problem to one where $i$ is the only player in the game.

In the next theorem, we prove that finding a best-response exactly or approximately is NP-hard. 

\begin{theorem}\label{thm:hardBestResp}
In the linear multi-activity model, both cost and budget, a player cannot approximate a best response beyond a constant factor in polynomial time unless {\em P = NP}. 
\end{theorem}

We note that in a single-player case finding a best response is equivalent to finding a PNE. We obtain the following corollary.

\begin{corollary}
In the linear multi-activity model, both cost and budget,
the problem of computing an exact or an approximate PNE is {\em NP}-hard. 
\end{corollary}


Theorem~\ref{thm:hardBestResp} proves that the problem of finding a best response in the linear multi-activity model does not admit a polynomial-time approximation scheme. For restricted versions of the model, a constant factor approximation can be found. For example, for the budget model, if the contests have $\{0,1\}$ weights for the activities, and the player has a budget that she can distribute across any of the activities (the hard instance constructed in Theorem~\ref{thm:hardBestResp} for the budget model satisfies these conditions), then the problem becomes a submodular maximization problem with a polynomial-time constant-factor approximation algorithm. However, we feel that a 
constant-factor approximation result is of limited usefulness in the context of
computing a best response or a Nash equilibrium.



Next, we study the fixed-parameter tractability of the problem of computing a best response. There are three natural parameters of the model: the number of players $n$, the number of contests $m$, and the number of activities $k$. We have already shown that the problem is NP-hard even with only one player, $n=1$. On the positive side, we show that the problem becomes tractable if either the number of contests $m$ or the number of activities $k$ is a constant.

\begin{theorem}\label{thm:fptBestResponse}
In the linear multi-activity model, both cost and budget, a player can compute a best response in polynomial time if either the number of contests or the number of activities is bounded by a constant. 
\end{theorem}
\section{Single-Activity Games: PPAD$\cap$PLS-Completeness }\label{sec:singleActivity}

In this section, we focus our attention on the single-activity model, and show that, for both cost and budget models, it is PLS-complete to find a pure Nash equilibrium and (PPAD$\cap$PLS)-complete to find a mixed Nash equilibrium (MNE). Note that the set of MNE of a game is a super-set of the set of PNE, so finding an MNE is at least as easy as finding a PNE.

In a single-activity game, for every player $i$ there is at most one activity $\ell$ for which the output $b_{i,\ell}$ may be strictly positive; for every other activity $\ell' \neq \ell$ it holds that $b_{i,\ell'} = 0$. Additionally, 
we assume that each contest $i$ has a fixed total prize, 
which it distributes equally among all winners.

\begin{theorem}\label{thm:ppadPls}
In the linear single-activity model, both cost and budget, it is {\sc (PPAD}\,$\cap$\,{\sc PLS)}-complete to find a mixed-strategy Nash equilibrium. 
\end{theorem}
A recent paper by Fearnley et al.~\cite{fearnley2021complexity} proved the interesting result that finding a fixed-point of a $2$-dimensional smooth function ($f : [0,1]^2 \rightarrow \bR$) given by a circuit of addition and multiplication operators using gradient descent, {\sc 2D-GD-FixedPoint}, is complete for the class PPAD$\cap$PLS. Based on this result, Babichenko and Rubinstein~\cite{babichenko2021settling} proved that the problem of computing an MNE of an explicit congestion game, {\sc ExpCong}, is also complete for PPAD$\cap$PLS. Babichenko and Rubinstein~\cite{babichenko2021settling} do this by first reducing {\sc 2D-GD-FixedPoint} to identical interest 5-polytensor games, {\sc 5-Polytensor}, and then {\sc 5-Polytensor} to {\sc ExpCong}. In our proof for Theorem~\ref{thm:ppadPls}, we reduce {\sc 5-Polytensor} to the single-activity game, proving the required result.

Before moving to the proof, let us define identical interest polytensor games, {\sc $\kappa$-Polytensor}. Polytensor games are a generalization of the better known polymatrix games, {\sc Polymatrix}; specifically, {\sc Polymatrix} = {\sc $2$-Polytensor}.
\begin{definition}[{\sc Polymatrix}: Identical Interest Polymatrix Game]
There are $n$ players. Player $i$ chooses from a finite set of actions $A_i$. The utility of player $i$ for action profile $(a_i, \ba_{-i}) = \ba = (a_j)_{j \in [n]} \in \times_{j \in [n]} A_j$ is given by $u_i(\ba) = \sum_{j \in [n], j \neq i} u_{i,j}(a_i, a_j)$, where $u_{i,j} : A_i \times A_j \rightarrow \bR_{\ge 0}$. The players have identical interest, i.e., $u_{i,j} = u_{j,i}$.
\end{definition}
The definition of {\sc $\kappa$-Polytensor} games is similar to that of {\sc Polymatrix} games: instead of $u_{i,j}$ we have $u_S$, where $S \subseteq [n], |S| = \kappa$.
\begin{definition}[{\sc $\kappa$-Polytensor}: Identical Interest $\kappa$-Polytensor Game]\label{def:polytensor}
There are $n$ players. Player $i$ chooses from a finite set of actions $A_i$. 
The utility of player $i$ for action profile $\ba = (a_j)_{j \in [n]} \in \times_{j \in [n]} A_j$ is given by $u_i(\ba) = \sum_{S \subseteq [n], i \in S, |S| = \kappa} u_S(\ba_S)$, where $\ba_S = (a_j)_{j \in S}$ is the action profile of the players in $S$ and $u_S : \times_{j \in S} A_j \rightarrow \bR_{\ge 0}$.
\end{definition}
When the number of actions for each player, $|A_i|$, is bounded by $m$, note that the representation size of a {\sc $\kappa$-Polytensor} game is $O(n^{\kappa}m^{\kappa})$. In particular, if $\kappa$ is a constant and
$m = poly(n)$, then {\sc $\kappa$-Polytensor} games admit a succinct representation.

\paragraph{Proof Sketch for Theorem~\ref{thm:ppadPls}.}
Below, we provide a reduction from {\sc $3$-Polytensor} to the budget game for a cleaner presentation of the main steps (unlike {\sc $5$-Polytensor}, {\sc $3$-Polytensor} is not known to be (PPAD$\cap$PLS)-complete). Similar steps with more calculations apply to {\sc $5$-Polytensor}, for both cost and budget models (we provide this argument in the full version of the paper).

Take an arbitrary instance of {\sc $3$-Polytensor} with $n$ players; we shall use the same notation as in Definition~\ref{def:polytensor}. We construct a single-activity game with $n$ players, $\sum_{i \in [n]} |A_i|$ activities, and a polynomial number of contests to be defined later.

The $\sum_{i \in [n]} |A_i|$ activities have a one-to-one association with the actions of the players. The activities are partitioned into $n$ subsets, so that the $i$-th subset has size $|A_i|$ and is associated with player $i$; we identify these activities with the set $A_i$. Player $i$ has a budget of $1$ that they can use to produce output along any activity from $A_i$, but they have $0$ budget for the activities in $A_j$ for $j \neq i$. Effectively, as we are in a single-activity model, player $i$ selects an activity from the activities in $A_i$ and produces an output of $1$ along it. Note that the players have disjoint sets of activities for which they can produce outputs.

All the contests we construct are associated with exactly three players and at most three activities. We shall denote a contest by $\cC_{i,j,k}(A)$, where (i) $S = \{i,j,k\}$ are the three distinct players whose utility function in the polytensor game, $u_S$, will be used to specify the prize of contest $\cC_S(A)$;
(ii) the contest $\cC_S(A)$ awards its prize to any player who produces an output of at least $1$ along the activities in $A$; (iii) the activities in $A$ are from $A_i \cup A_j \cup A_k$ with $|A| \le 3$ and $ |A \cap A_{\ell}| \le 1$ for $\ell \in S$.
We shall call a contest $\cC_S(A)$ a Type-$\ell$ contest if $|A| = \ell$.

Let us focus on a fixed set of three players $S = \{i,j,k\}$. We create contests to exactly replicate the utility that these players get from $u_S$. If we can do this, then, by repeating the same process for every triple of players, we will replicate the entire {\sc $3$-Polytensor} game. The utility that player $i$ gets from $u_S$ is $u_S(a_i,a_j,a_k)$, where $a_i$, $a_j$, and $a_k$ are the actions of the three players.
We have the following contests:

\textbf{Type-3 Contests.} Let us add a contest $\cC_S(a_i,a_j,a_k)$ with prize $v_S(a_i,a_j,a_k)$ for every $(a_i,a_j,a_k) \in A_i \times A_j \times A_k$. Later, we shall specify the $v_S(a_i,a_j,a_k)$ values based on $u_S(a_i,a_j,a_k)$ values. Contest $\cC_S(a_i,a_j,a_k)$ distributes a prize of $v_S(a_i,a_j,a_k)$ to players who produce output along the activities $a_i$, $a_j$, or $a_k$. 
    
Say, the players $i,j,k$ select the actions $a_i^*, a_j^*, a_k^*$. 
The total prize that player $i$ gets from the contests we added is:
\begin{align}\label{eq:ppadPls:1}
    &\frac{v_S(a_i^*,a_j^*,a_k^*)}{3}
    + \sum_{a_j \neq a_j^*} \frac{v_S(a_i^*,a_j,a_k^*)}{2}  + \sum_{a_k \neq a_k^*} \frac{v_S(a_i^*,a_j^*,a_k)}{2}
    + \sum_{\substack{a_j \neq a_j^* \\ a_k \neq a_k^*}} v_S(a_i^*,a_j,a_k)
\end{align}
In expression \eqref{eq:ppadPls:1}, the first term is for the prize that $i$ shares with $j$ and $k$, the second term is for the prizes that $i$ shares with $k$, but not with $j$, the third term is for the prizes that $i$ shares with $j$, but not with $k$, and the fourth term is for the prizes that $i$ does not share with $j$ or $k$. 

In expression \eqref{eq:ppadPls:1}, the first term $\frac{1}{3}{v_S(a_i^*,a_j^*,a_k^*)}$ resembles the utility  that the players obtain in the polytensor game, $u_S(a_i^*,a_j^*,a_k^*)$. If we were to set $v_S(a_i^*,a_j^*,a_k^*) = 3 u_S(a_i^*,a_j^*,a_k^*)$, then it would be exactly equal to it. However, we also need to take care of the additional terms in expression \eqref{eq:ppadPls:1}. Hence, we will add Type-1 and Type-2 contests to cancel these terms. 

The expression in \eqref{eq:ppadPls:1} can be rewritten as\\
    $\sum_{\substack{a_j \in A_j \\ a_k \in A_k}} v_S(a_i^*,a_j,a_k) - \sum_{a_j \neq a_j^*} \frac{v_S(a_i^*,a_j,a_k^*)}{2} - \sum_{a_k \neq a_k^*} \frac{v_S(a_i^*,a_j^*,a_k)}{2}
    -\frac{2v_S(a_i^*,a_j^*,a_k^*)}{3}$.
\\

\noindent\textbf{Type-1 Contests.} Let us add a contest $\cC_S(a_i')$ with prize
$$ 
\sum_{a_i \neq a_i', a_j \in A_j, a_k \in A_k} v_S(a_i,a_j,a_k)
$$ 
for every $a_i' \in A_i$. This contest $\cC_S(a_i')$ awards its prize to any player who produces output along activity $a_i'$ (effectively, it awards the prize to player $i$ if they produce output along $a_i'$, because no other player can produce output along $a_i'$). Similarly, we add the contests $\cC_S(a_j')$ and $\cC_S(a_k')$ for $a_j' \in A_j$ and $a_k' \in A_k$, respectively. The total prize that player $i$ gets from Type-1 and Type-3 contests is
$\sum_{a_i,a_j,a_k} v_S(a_i,a_j,a_k) - \sum_{a_j \neq a_j^*} \frac{1}{2}{v_S(a_i^*,a_j,a_k^*)} - \sum_{a_k \neq a_k^*} \frac{1}{2}{v_S(a_i^*,a_j^*,a_k)}
    -\frac{2}{3}{v_S(a_i^*,a_j^*,a_k^*)}$.
As $\sum_{a_i,a_j,a_k} v_S(a_i,a_j,a_k)$ does not depend upon the action $a_i^*$ selected by $i$, the utility of player $i$ is effectively
\begin{equation*}
    - \sum_{a_j \neq a_j^*} \frac{1}{2}{v_S(a_i^*,a_j,a_k^*)} - \sum_{a_k \neq a_k^*} \frac{1}{2}{v_S(a_i^*,a_j^*,a_k)} -\frac{2}{3}{v_S(a_i^*,a_j^*,a_k^*)}.
\end{equation*}

\noindent\textbf{Type-2 Contests.} Let us add a contest $\cC_S(a_i', a_j')$ with prize
$$
\sum_{a_k \in A_k} \frac{1}{2}{v_S(a_i',a_j',a_k)}
$$
for every $a_i' \in A_i$ and $a_j' \in A_j$. This contest $\cC_S(a_i', a_j')$ awards its prize to players who produce output along activity $a_i'$ or $a_j'$. In a similar manner, we add contests corresponding to the actions of the other $5$ possible combinations of players among the three players $i,j,k$, e.g., $\cC_S(a_i', a_k')$ for $a_i' \in A_i$ and $a_k' \in A_k$, and so on. The net utility that player $i$ gets from Type-1, Type-2 and Type-3 contests is
    $\frac{1}{3}{v_S(a_i^*,a_j^*,a_k^*)}$.
We set $v_S(a_i,a_j,a_k) = 3 u_S(a_i,a_j,a_k)$ for every $(a_i, a_j, a_k) \in A_i \times A_j \times A_k$, and we are done.
\qed

In Theorem~\ref{thm:ppadPls}, we analyzed the complexity of computing an MNE. As we have shown earlier, single-activity cost and budget games always have a PNE, and therefore, it is relevant to know the complexity of computing a PNE. 
In the next theorem, we prove that computing a PNE is PLS-complete, and this is true even if all the players are identical. In the proof, we reduce {\sc Max-Cut} to the problem of finding a PNE in a particular class of single-activity games with identical players. For this class of single-activity games, finding an MNE is easy, which highlights that the class of single-activity games where PNE is hard (PLS-complete) to compute is strictly larger than the class of single-activity games where MNE is hard (PPAD$\cap$PLS-complete) to compute. 
\begin{theorem}\label{thm:plsCompleteness}
In the single-activity models, both cost and budget, it is {\em PLS}-complete to find a pure-strategy Nash equilibrium. The result holds even if all players are identical.
\end{theorem}

A direct corollary of this PLS-completeness result is that better/best-response dynamics takes an exponential number of steps to converge for some instances of the problem~\cite{johnson1988easy}.

Regarding fixed-parameter tractability (FPT), in the full version of the paper we show that a PNE is efficiently computable for both cost and budget single-activity models if the number of players is a constant, and for the budget model if the number of contests is a constant. Providing FPT results for other cases remains open.

\section{Conclusion and Future Work}
In this paper, we studied a model of simultaneous contests and analyzed the existence, efficiency, and computational complexity of equilibria in these contests. 
Given the real-life relevance of the three-level model, player--activity--contest, it will be interesting to study it for prize allocation rules other than the equal-sharing allocation, such as rank-based allocation, proportional allocation, etc. For these contests, one may investigate the properties of the equilibria and their computational complexity. We also believe that there is much to explore regarding the computational complexity of simultaneous contests, in general.

\section*{Acknowledgements}
We thank the anonymous reviewers for their valuable feedback. The second author is supported by Clarendon Fund and SKP Scholarship.

\bibliographystyle{splncs04}
\bibliography{ref}

\appendix

\section{Omitted Proofs and Additional Results}
\subsection{From Section~\ref{sec:modelpnepoa}}
\begin{proof}[Theorem~\ref{thm:multPNE}]
We shall prove the result for the cost game. The budget game can be reduced to the cost game by endowing it with cost function $c(\cdot)$ that is always equal to $0$ (the presence of the budget constraint does not affect the proof).

Let player $i$ make a move by changing its strategy from $\bb_i$ to $\bb_i'$, and let $\bb_{-i}$ be the pure strategy profile of the other players. Let $S \subseteq M$ and $S' \subseteq M$ be the contests that player $i$ wins for strategy profiles $\bb = (\bb_i, \bb_{-i})$ and $(\bb_i', \bb_{-i})$, respectively. 
We have 
$$
n_j(\bb_i', \bb_{-i}) = 
\begin{cases}
n_j(\bb) - 1 \quad\text{ if $j \in S \setminus S'$}\\
n_j(\bb) + 1 \quad\text{ if $j \in S' \setminus S$}\\
n_j(\bb) \quad\text{ in all other cases}.
\end{cases}
$$
The change in the potential $\phi(\bb_i', \bb_{-i}) - \phi(\bb)$ can be expressed as
\begin{align*}
    &\  \sum_{j \in M} \left( \sum_{\ell \in [n_j(\bb_i', \bb_{-i})]} v_j(\ell) - \sum_{\ell \in [n_j(\bb)]} v_j(\ell) \right) - \left(  c_{i}(\bb_i') - c_i(\bb_i)  \right) \\
    & \qquad \qquad \qquad \qquad \qquad \qquad \qquad \qquad \qquad + \sum_{q \in N \setminus \{i\}}  \underbrace{\left(  c_{q}(\bb_{q}) -  c_{q}(\bb_{q}) \right)}_{= 0}, \\
    &= \underbrace{\left(\sum_{j \in S' \setminus S} v_j(n_j(\bb_i', \bb_{-i})) -  c_i(\bb'_i)\right)}_{u_i(\bb_i', \bb_{-i})}  - \underbrace{\left( \sum_{j \in S \setminus S'}v_j(n_j(\bb)) - c_i(\bb_i) \right)}_{u_i(\bb)}\ ,
\end{align*}
which is exactly equal to the change in the utility of the player $i$.
\qed\end{proof}

\begin{proof}[Corollary~\ref{thm:pseudoPolyBR}]
The potential function is bounded from above by $\phi(\bb) \le \sum_{j \in M} \sum_{\ell \in [n_j(\bb)]} v_j(\ell) \le \sum_{j \in M} \sum_{\ell \in [n_j(\bb)]} v_j(1) \le \sum_{j \in M} n \cdot v_j(1) \le n \cdot m \cdot (\max_{j} v_j(1))$. As the game is an exact potential game, every $\epsilon$-better-response move increases the potential by more than $\epsilon$. So, in $n \cdot m \cdot (\max_{j} v_j(1)) / \epsilon$ steps the sequence should converge to an $\epsilon$-PNE.
\qed\end{proof}

\begin{proof}[Lemma~\ref{thm:PoA:1}]
Let $S, T \subseteq \cG$, $S \subseteq T$, and $x \in \cG \setminus T$. We need to show that 
\[
    \sw(S \cup \{x\}) - \sw(S) \ge \sw(T \cup \{x\}) - \sw(T).
\]
Let $x = g^i_j$ for some player $i$ and contest $j$. As $S \subseteq T$, $n_j(S) \le n_j(T)$. Now, as the total prize allocated by contest $j$ (which is equal to $\ell \cdot v_j(\ell)$ if there are $\ell$ winners) is a concave function, we get
\begin{multline*}
    (n_j(S) + 1) \cdot v_j(n_j(S) + 1) - n_j(S) \cdot v_j(n_j(S)) \\
    \ge (n_j(T) + 1) \cdot v_j(n_j(T) + 1) - n_j(T) \cdot v_j(n_j(T)).
\end{multline*}
As $n_{j}(S \cup \{x\}) = n_{j}(S) + 1$ and $n_{j}(T \cup \{x\}) = n_{j}(T) + 1$, and $n_{j'}(S \cup \{x\}) = n_{j'}(S)$ and $n_{j'}(T \cup \{x\}) = n_{j'}(T)$ for $j' \neq j$, we get the required result.
\qed\end{proof}

\begin{proof}[Lemma~\ref{thm:PoA:2}]
By definition of $u_i$,
\begin{equation}\label{eq:prf:PoA:1}
    u_i(A_i, A_{-i}) = \sum_{j : g^i_j \in A_i} v_j( n_j(A_i, A_{-i}) ) = \sum_{j : g^i_j \in A_i} v_j( n_j(\emptyset, A_{-i}) + 1 ).
\end{equation}
By changing her action from $\emptyset$ to $A_i$, player $i$ starts winning contests in $J = \{j \mid g^i_j \in A_i\}$. Therefore, she affects the prize allocation of players who were already receiving prizes from contests in $J$. In particular, for any $j \in J$, the $n_j(\emptyset, A_{-i})$ players who were receiving a prize of $v_j( n_j(\emptyset, A_{-i}) )$ now receive a prize of $v_j( n_j(\emptyset, A_{-i}) + 1 )$. So,
\begin{multline}\label{eq:prf:PoA:2}
    \sum_{\iota \neq i} u_{\iota}(\emptyset, A_{-i}) - \sum_{\iota \neq i} u_{\iota}(A_i, A_{-i}) \\
    = \sum_{j \mid g^i_j \in A_i} n_j(\emptyset, A_{-i}) \cdot (v_j( n_j(\emptyset, A_{-i}) ) - v_j( n_j(\emptyset, A_{-i}) + 1 )).
\end{multline}
As the total prize allocated by a contest is a non-decreasing function of the number of winners, we have
\begin{align*}
     &(n_j(\emptyset, A_{-i}) + 1) \cdot v_j( n_j(\emptyset, A_{-i}) + 1 ) \ge n_j(\emptyset, A_{-i}) \cdot v_j( n_j(\emptyset, A_{-i}) ) \\
     &\implies v_j( n_j(\emptyset, A_{-i}) + 1 ) \ge n_j(\emptyset, A_{-i}) \cdot (v_j( n_j(\emptyset, A_{-i}) ) - v_j( n_j(\emptyset, A_{-i}) + 1 )) \\
     &\implies \sum_{j : g^i_j \in A_i} v_j( n_j(\emptyset, A_{-i}) + 1 ) \\
     &\qquad \qquad \qquad \qquad \ge \sum_{j : g^i_j \in A_i} n_j(\emptyset, A_{-i}) \cdot (v_j( n_j(\emptyset, A_{-i}) ) - v_j( n_j(\emptyset, A_{-i}) + 1 )).
\end{align*}
Combining this inequality with \eqref{eq:prf:PoA:1} and \eqref{eq:prf:PoA:2}, we get the required result.
\qed\end{proof}

\begin{proof}[Theorem~\ref{thm:swBudMul}]
The social welfare is a non-decreasing function, as the total prize allocated by a contest weakly increases as the number of winners increases. This observation, together with Lemmas~\ref{thm:PoA:1} and~\ref{thm:PoA:2}, shows that our setting satisfies the necessary conditions of Theorem~5 of
Vetta~\cite{vetta2002nash}. This gives us the required bound on the price of anarchy.
\qed\end{proof}

\begin{proof}[Theorem~\ref{thm:swBudMulNeg}]
We describe an instance of the game, adapted from Example 5.37 of Vojnovic~\cite{vojnovic2015contest}, that proves the required result.

Let there be $n$ players, $n$ activities, and $n$ contests, i.e., $m = k = n$. Each player has a budget of $1$ that he can spend on any of the activities. Contest $1$ has a prize of $n+1$ and the remaining $n-1$ contests have a prize of $1$ each. Contest $i$ gives (a share of) its prize to any player whose output is at least $1$ for activity $i$.

This game has a unique equilibrium where every player produces an output of $1$ for activity $1$. This happens because even if all players play action $1$, each player gets a share of the prize from contest $1$, which is equal to $(n+1)/n > 1$. On the other hand, by spending their budget on any other activity they can get a prize of at most $1$.

Now, the social welfare for this equilibrium is $n+1$. On the other hand, if each player selects a different activity then the social welfare would be $2n$. So, we get a price of stability equal to $\frac{2n}{n+1} \rightarrow 2$ as $n \rightarrow \infty$.
\qed\end{proof}

\begin{proof}[Theorem~\ref{thm:swCost}]
We analyze the two variants of the model separately.

\smallskip

\noindent\textbf{Subtracting Costs from Social Welfare.}
We will construct an instance with two players, two activities, and two contests, i.e., $n=k=m=2$. Contest $j$ distributes a prize of $2+\epsilon$ to players who produce an output of $1$ along activity $j$ and does not care about the output along activity $(3-j)$. Each player has a linear cost function with a slope of $1$ for each activity.

This game has a unique equilibrium where both players produce an output of $1$ for both activities. Because, irrespective of what the other player is doing, a player would get a prize of strictly greater than $1$ by producing an output of $1$ (with cost $1$) for each activity. The total welfare for this equilibrium is $2*(2+\epsilon) - 2*(1+1) = 2\epsilon$.


Consider a socially optimal situation where player $i$ produces output along activity $i$: the total prize allocation is still $2*(2+\epsilon)$, but the total cost is $2$, so the total welfare is $2+2\epsilon$. Therefore, the price of anarchy is $\frac{2+2\epsilon}{2\epsilon}$, which tends to $+\infty$ as $\epsilon \rightarrow 0$.

\smallskip

\noindent\textbf{Not Subtracting Costs from Social Welfare.} In this definition of social welfare, too, a simple example shows that the price of anarchy is unbounded. Consider an instance with one player, two activities, and two contests. The two activities have a one-to-one correspondence with the contests, i.e., contest $j$ cares only about activity $j$. The player can win the first contest, which has a prize of $\epsilon$, with a cost of $\epsilon/2$, and she can win the second contest, which has a prize of $1$, with a cost of $1+\epsilon/2$. In equilibrium (unique), the player only wins the first contest. The PoS is $1/\epsilon$, which tends to $+\infty$ as $\epsilon \rightarrow 0$. 
Here, the social welfare and the utility of the player are not aligned, as the player cares about costs whereas the social welfare does not. Hence a bad PoS is expected.
\qed\end{proof}

\subsection{From Section~\ref{sec:multiActivity}}
\begin{proof}[Theorem~\ref{thm:hardBestResp}]
The following two lemmas prove the result separately for the cost and the budget models.

\begin{lemma}\label{thm:hardApproxBestRespCost}
In the linear multi-activity cost model, there is no polynomial-time approximation scheme for best-response, unless {\em P = NP}.
\end{lemma}
\begin{proof}[Lemma~\ref{thm:hardApproxBestRespCost}]
We provide a reduction from {\sc 4-Regular-SetCover}, which is NP-hard to approximate by a factor of $2-\epsilon$ for any $\epsilon > 0$~\cite{holmerin2002vertex}. 

{\sc 4-Regular-SetCover}: An instance of this problem is given by a universe $\cU = \{1, 2, \ldots, \hat{m}\}$, a set $\cS = \{S_1, S_2, \ldots, S_{\hat{k}}\}$, where each $S_j$ is a subset of $\cU$ of size $4$ ($|S_j|=4$). The objective is to find a set cover of $\cU$ of minimum possible size, i.e., find a collection of sets $\cS'\subseteq \cS$ with minimum possible $|\cS'|$ such that $\cup_{S_j\in \cS'}=\cU$.
We can assume without loss of generality that 
$\cup_{S_j\in \cS}=\cU$ (as otherwise we obviously cannot have a set cover).

We create an instance of the multi-activity cost model with $1$ player, $m = \hat{k} + \hat{m}$ contests, and a set of $k=\hat{k}$ activities $K$. We split the contests into two groups of size $\hat{k}$ and $\hat{m}$, respectively:  $M_1 = \{1, \ldots, \hat{k}\}$ and $M_2 = \{\hat{k}+1, \ldots, \hat{k}+\hat{m}\}$.
\begin{itemize}
    \item Each contest in $M_1$ has a prize of value $(\hat{m} + 1)$. Let the weights for these contests be defined as:
    \[
        w_{j,\ell} = \begin{cases} 1,&\text{ for $\ell = j$}\\
        0,&\text{ otherwise}\end{cases}
    \]
    where $j \in M_1$ and $\ell \in K$. That is, a contest $j \in M_1$ awards a prize of $(\hat{m} + 1)$ to the player as long as she produces an output of at least $1$ along activity $j$.
    \item Each contest in $M_2$ has a prize of value $1+1/(\hat{m}+1)$. To receive the prize of contest $(j+\hat{k}) \in M_2$, the player needs to produce a total output of $1$ along activities $\{ \ell \mid j \in S_{\ell} \}$. Formally, the weights are:
    \[
        w_{j,\ell} = \begin{cases} 1,&\text{ if $ (j-\hat{k}) \in S_{\ell}$}\\
        0,&\text{ otherwise}\end{cases}
    \]
    where $j \in M_2$ and $\ell \in K$.
    \item The player has a marginal cost of $(\hat{m} + 2)$ for producing output along any of the activities, i.e., $c_{1,\ell} = \hat{m} +2$ for $\ell \in K$.
\end{itemize}

Observe that the contests in $M_1$ are much more valuable to the player than the ones in $M_2$: each contest in $M_1$ has a prize of $(\hat{m}+1)$, which is more than the total prize of all contests in $M_2$ which is equal to $\hat{m} + \hat{m}/(\hat{m}+1)$.
Note that it is sub-optimal for the player to produce an output that is strictly more than $1$ along any activity: this would contribute to her cost without affecting the prizes she gets. Also, it is sub-optimal for a player to produce an output that is strictly between $0$ and $1$ for any activity. Indeed, let 
$K'=\{\ell\in K\mid 0<b_{1, \ell}<1\}$, and set
$\beta=\sum_{\ell\in K}b_{1, \ell}$. Suppose for the sake of contradiction that $K'\neq\emptyset$. The activities in $K'$ do not help the player to win contests in $M_1$. If $\beta< 1$, they do not help the player to win contests in $M_2$ either, so the player can improve her utility by 
setting $b_{1, \ell}=0$ for each $\ell\in K'$.
If $\beta\ge 1$, these efforts may contribute towards winning contests in $M_2$, whose total prize value is $\hat{m} + \hat{m}/(\hat{m}+1)$, but then they incur a cost of at least $\beta\cdot(\hat{m} + 2) \ge \hat{m} + 2 > \hat{m} + \hat{m}/(\hat{m}+1)$, so again the player can improve her utility by 
setting $b_{1, \ell}=0$ for each $\ell\in K'$.
 Hence, when playing a best response, the player produces binary outputs (either $0$ or $1$) along each activity. In other words, she selects a subset of activities from $K$ and produces an output of $1$ for each.

Without loss of generality, we can also assume that the player will select an action (subset of activities) that wins her all the contests in $M_2$. For contradiction, say the player selects an action $A \subset K$ such that there is a contest $(j+\hat{k}) \in M_2$ that the player does not win. As per our assumption that every element of the universe $\cU$ in the {\sc 4-Regular-SetCover} instance is covered by at least one set in $\cS$, so there is an $\ell \in [\hat{k}]$ such that $j \in S_{\ell}$. $\ell \notin A$ because the player does not win contest $j+\hat{k}$. Now, if the player selects action $A'= \{\ell\} \cup A$, i.e., produces an output of $1$ along activity $\ell$ in addition to the ones in $A$, then her prize increases by $(\hat{m}+1)$ because of contest $\ell \in M_1$ and by $1 + 1/(\hat{m}+1)$ because of contest $(j+\hat{k}) \in M_2$, while her cost increases by $(\hat{m}+2)$. So, her utility for playing $A'$ is strictly more than $A$. 

In summary, given any action by a player, we can can convert this action to another action where the player produces an output of $1$ for a subset of activities $A \subseteq K$, produces an output of $0$ for activities in $K \setminus A$, and wins all the contests in $M_2$. And, this conversion increases the player's utility and can be done in polynomial time. Further, this new action directly corresponds to a set cover for the {\sc 4-Regular-SetCover} instance.

We claim that the player cannot compute an action $A \subseteq K$ in polynomial time that gives her a utility of at least $(3+\epsilon)/4$ of the optimal utility, for any $\epsilon > 0$. As per our discussion above, w.l.o.g., $A$ is a set cover, i.e., $\cup_{\ell \in A}S_{\ell} = \cU$. Let $A^* \subseteq [\hat{k}]$ be an optimal set cover. The utility of the player for action $A$ is $\hat{m} + \frac{\hat{m}}{\hat{m}+1} - |A|$ while the utility for action $A^*$ is $\hat{m} + \frac{\hat{m}}{\hat{m}+1} - |A^*|$. If the utility for playing $A$ is at least $(3+\epsilon)/4$ of the optimal utility, then
\begin{multline}\label{eq:costBestRes1}
    \hat{m} + \frac{\hat{m}}{\hat{m}+1} - |A| \ge \frac{3+\epsilon}{4} \left( \hat{m} + \frac{\hat{m}}{\hat{m}+1} - |A^*|  \right) \\
    \implies \frac{1-\epsilon}{4} \left( \hat{m} + \frac{\hat{m}}{\hat{m}+1}\right) + \frac{3+\epsilon}{4} |A^*| \ge |A|.
\end{multline}
Now, as $|S| = 4$ for every $S \in \cS$, therefore 
\[
    |A^*| \ge \frac{|\cU|}{4} = \frac{\hat{m}}{4} \implies 4|A^*| \ge \hat{m} \implies 5|A^*| \ge \hat{m} + 1 \ge \hat{m} + \frac{\hat{m}}{\hat{m}+1}.
\]
Putting this in \eqref{eq:costBestRes1} gives us
\begin{equation*}
    \frac{1-\epsilon}{4} 5|A^*| + \frac{3+\epsilon}{4} |A^*| \ge |A| \implies |A| \le (2-\epsilon)|A^*|.
\end{equation*}
So, $A$ is a $2-\epsilon$ approximation of the optimal set cover $A^*$ for the arbitrary {\sc 4-Regular-SetCover} instance, which is not possible unless $P = NP$~\cite{holmerin2002vertex}.

\qed\end{proof}

\begin{lemma}\label{thm:hardApproxBestRespBudget}
In the linear multi-activity budget model, there is no polynomial-time approximation scheme for best-response, unless {\em P = NP}.
\end{lemma}
\begin{proof}[Lemma~\ref{thm:hardApproxBestRespBudget}]
We prove this result by providing a reduction from \textsc{MAX-2-SAT-3}. In an instance of the \textsc{MAX-2-SAT-3} problem, we are given a 2-CNF formula where each variable occurs at most $3$ times, and the objective is to find an assignment of the variables to satisfy as many clauses of the 2-CNF formula as possible. It is known that it is NP-hard to approximate this problem beyond a constant factor~\cite{berman1999some}. 

Consider an instance of \textsc{MAX-2-SAT-3} with $\hat{n}$ variables $x_1, x_2, \ldots, x_{\hat{n}}$. Let there be $\hat{m}$ clauses in the 2-CNF formula; w.l.o.g. assume that $\hat{m} \ge \hat{n}$. We construct an instance of the multi-activity budget game with one player, $k = 2 \hat{n}$ activities, and $m = \hat{m} + 4 \hat{n}$ contests. The player has a budget of $\hat{n}$ that she can allocate across any of the $2\hat{n}$ activities. Each contest awards a prize of $1$.
\begin{itemize}
    \item Activities: Corresponding to each variable in the 2-CNF formula, we have two activities in the game. For variable $x_i$, we have activity $i$ with output $b_i$, corresponding to the literal $x_i$, and we have activity $\hat{n}+i$ with output $b_{\hat{n}+i}$, corresponding to the literal $\neg x_i$. We sometimes denote $b_{\hat{n}+i}$ by $b_i'$ for a cleaner presentation.
    \item Contests: The $m = \hat{m} + 4 \hat{n}$ contests are of two types:
    \begin{enumerate}
        \item Contests $M_1$. There are $\hat{m}$ contests of this type, and these contests correspond to the clauses in the 2-CNF formula. For a clause $x_i \land x_j$, we have a contest with condition $b_i + b_j \ge 1$. Similarly, for a clause $x_i \land \neg x_j$, there is a contest with condition $b_i + b_j' \ge 1$, and so on for the other types of clauses. A contest in $M_1$ can be identified by the associated pair of activities $(i,j) \in [2 \hat{n}] \times [2 \hat{n}]$.
        \item Contests $M_2$. There are $4\hat{n}$ contests of this type, and they correspond to the $2\hat{n}$ possible literals in the 2-CNF formula. For each $i \in [\hat{n}]$, there are two contests with condition $b_i \ge 0$ and two with condition $b_i' \ge 0$.
    \end{enumerate}
\end{itemize}
As each contest has a prize of $1$, the objective of the player is to satisfy as many contests as possible.

Let us assume that for linear multi-activity budget games, we have a polynomial-time algorithm that approximates the best response by a factor $\gamma \in (0,1]$.
Let $\bb$ be a solution for the instance of the problem we just constructed, found in polynomial time using this approximation algorithm.

From this approximate solution $\bb$, we construct a $\{0,1\}$-integral approximate solution with at least as good an approximation ratio as $\bb$, as described below:
\begin{enumerate}
    \item If there is an $i \in [2 \hat{n}]$ such that $b_i > 1$, set $b_i = 1$. Notice that this transformation does not affect the number of satisfied contests. In the subsequent transformations of $\bb$, we will not touch these $b_i$ values (where $b_i = 1$) and the contests that contain them. So, the contests where they show up will remain satisfied. Let these contests be denoted by $M^{(1)}$.
    \item Let $M_1^{(2)}$ denote the contests in $M_1 \setminus M^{(1)}$ that are satisfied by $b$. Do the following updates to $\bb$ and $M_1^{(2)}$ until  $M_1^{(2)} = \emptyset$:
    \begin{itemize}
        \item Take an arbitrary contest $(i,j)$ in $M_1^{(2)}$. We know that the following conditions hold: $b_i + b_j \ge 1$, $b_i < 1$, and $b_j < 1$. 
        \item Set $b_i = 1$ and $b_j = 0$. As the value of $b_i$ increases, the contests that contain $b_i$ and were satisfied before this update remain satisfied after this update. On the other hand, as every variable may occur in at most three clauses, $b_j$ may occur in at most two more contests in $M_1^{(2)}$ in addition to the current contest $(i,j)$. These contests become unsatisfied as a result of this transformation, but they are compensated by the two contests in $M_2$ that correspond to $b_i$, which were unsatisfied before this update, but become satisfied after this update.
        \item Remove all contests from $M_1^{(2)}$ that correspond to the activities $i$ or $j$, i.e., set $M_1^{(2)} = M_1^{(2)} \setminus \{ (\ell, t) \mid  \{\ell, t\} \cap \{i,j\} \neq \emptyset \}$.
    \end{itemize}
    As the size of $M_1^{(2)}$ decreases after each update without decreasing the total number of satisfied contests, the process terminates with a solution as good as before.
    \item As a result of the previous step, there is no contest in $M_1$ that is satisfied but has an output strictly less than $1$ for both of its activities. In other words, for any contest $(i,j) \in M_1$, $b_i + b_j \ge 1 \implies \max(b_i, b_j) = 1$. Moreover, none of the contests in $M_2$ can be satisfied by an activity $i$ with $b_i < 1$. So, we can safely set $b_i = 0$ for all $i$ with $b_i < 1$, without decreasing the number of satisfied contests.
\end{enumerate}
Thus, we can assume without loss of generality that the solution for the instance of the linear multi-activity budget game is $\{0,1\}$-integral. We can also assume that $\sum_{i \in 2 \hat{n}} b_i = \hat{n}$, because setting additional $b_i$ values to $1$ does not decrease the objective. 

As mentioned before, \textsc{MAX-2-SAT-3} is NP-hard to approximate beyond a constant factor, even when the number of clauses is of the same order as the number of variables, i.e., $\hat{m} = \Theta(\hat{n})$, as proved by Berman and Karpinski~\cite{berman1999some}, given in the theorem below. 

\begin{theorem}\cite{berman1999some}\label{thm:bk}
For any $0 < \epsilon < 1/2$, it is {\em NP}-hard to decide whether an instance of {\sc MAX-2-SAT-3} with $\hat{m} = \mu_1 \hat{n}$ clauses has a truth assignment that satisfies $ (\mu_2 + 1 - \epsilon) \hat{n}$ clauses, or every truth assignment satisfies at most $(\mu_2 + \epsilon) \hat{n}$ clauses, where $\mu_1$ and $\mu_2$ are two positive constants.\footnote{$\mu_1 = 2016$ and $\mu_2 = 2011$. Berman and Karpinski~\cite{berman1999some} call the problem \textsc{3-OCC-MAX-2SAT} instead of \textsc{MAX-2-SAT-3}.}
\end{theorem}
Theorem~\ref{thm:bk} proves that \textsc{MAX-2-SAT-3} cannot be approximated better than a factor $(\mu_2 + 1)/\mu_2 - \epsilon$. Now, if an instance of the $\textsc{MAX-2-SAT-3}$ problem has a truth assignment that satisfies $\gamma \hat{n}$ clauses, then the constructed instance of the linear budget game has a best response that satisfies $(\gamma + 4) \hat{n}$ contests. So, for any $0 < \epsilon < 1/2$, it is NP-hard to decide whether an instance of the linear budget game with $m = \hat{m} + 4 \hat{n} = (\mu_1 + 4)\hat{n}$ contests has a best response that satisfies $ (\mu_2 + 5 - \epsilon) \hat{n}$ contests, or it can be at most $(\mu_2 + 4 + \epsilon) \hat{n}$. This gives us an approximation ratio of $(\mu_2 + 5)/(\mu_2 + 4) - \epsilon$.


\qed\end{proof}
\qed\end{proof}

\begin{proof}[Theorem~\ref{thm:fptBestResponse}]
We prove the result first for the case when the number of activities is a constant and then for the case when the number of players is a constant.

\textbf{Constant number of activities ($k$ constant).} Notice that in the linear multi-activity model, the criterion for a given contest $j$ is a linear constraint of the form $\sum_{\ell \in K} w_{j,\ell} b_{i,\ell} \ge 1$, where $w_{j,\ell} \in \bR_{\ge 0}$ is a non-negative weight that contest $j$ has for activity $\ell$ and $\bb_i = (b_{i,\ell})_{\ell \in K}$ is the output vector for a player. This linear constraint divides $\bR^k$ into two half-spaces. The set of half-spaces in $\bR^k$ has a VC-dimension of $k+1$~\cite{kearns1994introduction}, and therefore, the set of constraints (one for each contests) also has a VC-dimension of at most $k+1$. Let $\cS$ be the following set of subsets of contests:
\begin{multline*}
    \cS = \{ S \subseteq M \mid \exists \bb_i \in \bR_{\ge 0}^k \text{ s.t. $\sum_{\ell \in K} w_{j,\ell} b_{i,\ell} \ge 1, \forall j \in S$,} \\
    \text{and $\sum_{\ell \in K} w_{j,\ell} b_{i,\ell} < 1, \forall j \in S$} \}.
\end{multline*}
In other words, for every set $S \in \cS$, there exists an action $\bb_i$ which satisfies the criteria for all the contests in $S$ and none of the contests not in $S$. As the set of criteria corresponding to the contests has a VC-dimension of $k+1$, the set $\cS$ is of size $O(m^{k+1})$ and the elements of $\cS$ can be enumerated in time polynomial in $m$ when $k$ is fixed.\footnote{See Lemma~3.1 and 3.2 of Chapter 3 of the book by Kearns and Vazirani~\cite{kearns1994introduction}.} The elements of $\cS$ can be computed inductively, see the proof of Lemma~3.1 of Kearns and Vazirani~\cite{kearns1994introduction}. 

Now, let us pick a set $S \in \cS$. For the cost model, we can efficiently solve a linear program where the objective is to minimize the (linear) cost function of the player and the constraints are the ones corresponding to the contests in set $S$. Similarly, for the budget model, we need to check the feasibility of the set of constraints for contests in $S$ plus the budget constraint of the player, which can also be done efficiently. Finally, we compile the optimal solution for each $S \in \cS$ to find the optimal solution overall.

\textbf{Constant number of contests ($m$ constant).} The proof is similar to the one for the case when $k$ is a constant. There are $2^m$ possible subsets of contests, we can efficiently enumerate over them because $m$ is a constant. For each subset of contests $S \subseteq M$, we repeat the same process (solve a linear programming problem) as the previous case when $k$ was a constant.
\qed\end{proof}

\subsection{From Section~\ref{sec:singleActivity}}
\begin{proof}[Theorem~\ref{thm:ppadPls}]
For now, let us focus on the single-activity budget model. At the end, we extend the result to the cost model.

In the proof sketch given in the main paper, we reduced a {\sc $3$-Polytensor} game to a single-activity budget game to emphasize the main ideas and give a cleaner presentation, but {\sc $3$-Polytensor} games have not yet been proven to be (PPAD$\cap$PLS)-complete. Rather, {\sc $5$-Polytensor} games are (PPAD$\cap$PLS)-complete~\cite{babichenko2021settling}, so we start from an instance of {\sc $5$-Polytensor}. 

Take an arbitrary instance of {\sc $5$-Polytensor} with $n$ players; we shall use the same notation as Definition~\ref{def:polytensor}. We construct a single-activity game with $n$ players, $\sum_{i \in [n]} |A_i|$ activities, and a polynomial number of contests to be defined later.

The $\sum_{i \in [n]} |A_i|$ activities have a one-to-one association with the actions of the players. The activities are partitioned into $n$ subsets; the $i$-th subset has size $|A_i|$ and is associated with player $i$; we identify these activities by $A_i$. Player $i$ has a budget of $1$ that they can use to produce output along any activity from $A_i$, but they have $0$ budget for the activities in $A_j$ for $j \neq i$. Effectively, as we are in a single-activity model, player $i$ selects an activity from the activities in $A_i$ and produces an output of $1$ along it. Note that the players have disjoint sets of activities for which they can produce outputs.

All the contests we construct are associated with exactly five players and at most five activities. We shall denote a contest by $\cC_S(A)$, where
\begin{itemize}
    \item $S$ is the set of five distinct players based on whose joint utility function in the polytensor game, $u_S$, we shall specify the prize of contest $\cC_S(A)$;
    \item the contest $\cC_S(A)$ awards its prize to any player who produces an output of at least $1$ along the activities in $A$;
    \item the activities in $A$ are from $\cup_{i \in S} A_i $ with $|A| \le 5$ and $ |A \cap A_i| \le 1$ for $i \in S$.
\end{itemize}
We shall call a contest $\cC_S(A)$ a Type-$\ell$ contest if $|A| = \ell$.

We create the contests $(\cC_S(A))_A$ to exactly replicate the utility that the five players in $S$ get from $u_S$. If we can do this, then repeating the same process for every set of five players, we will replicate the entire {\sc $5$-Polytensor} game. The utility that player $i \in S$ gets from $u_S$ is $u_S(\ba_S)$, where the $\ba_S$ are the actions of the five players.
We have the following contests:

\textbf{Type-5 Contests.} Let us add a contest named $\cC_S(\ba_S)$ with prize $v_S(\ba_S)$ for every $\ba_S \in \times_{i \in S} A_i$. We shall specify the $v_S(\ba_S)$ values based on the $u_S(\ba_S)$ values, later. Contest $\cC_S(\ba_S)$ distributes its prize to players who produce output along the activities $\ba_S$. 
    
Say the players $S$ play the actions $\ba_S^* = (a_j^*)_{j \in S}$. 
Let $A_j^- = A_j \setminus \{a_j^*\}$; and for $T \subseteq S$, let $A_T = \times_{j \in T} A_j$ and $A_T^- = \times_{j \in T} A_j^-$. 
The total prize that player $i \in S$ gets from the Type-5 contests is:
\begin{align*}
    \sum_{\ell \in [5]} \sum_{T \subseteq S, i \in T, |T| = \ell} \sum_{\ba_{S \setminus T} \in A_{S \setminus T}^-} \frac{v_S(\ba_T^*, \ba_{S \setminus T})}{\ell} .
\end{align*}
In the above formula, the outer summation with $\ell$ is for the number of players that player $i$ has to share their prize with, including $i$; the middle summation is for the $\ell$ players, $T$; the inner summation is for the actions of the players in $S \setminus T$. The same formula can be rewritten as
\begin{align*}
    &\sum_{\ba_{S \setminus \{i\}} \in A_{S \setminus \{i\}}} v_S(a_i^*, \ba_{S \setminus \{i\}}) -\sum_{2 \le \ell \le 5} \sum_{T \subseteq S, i \in T, |T| = \ell} \sum_{\ba_{S \setminus T} \in A_{S \setminus T}^-} \frac{\ell-1}{\ell} v_S(\ba_T^*, \ba_{S \setminus T}).
\end{align*}

\textbf{Type-1 Contests.} Let us add a contest named $\cC_S(a_i')$ with prize\\
$\sum_{a_i \neq a_i', \ba_{S \setminus \{i\}} \in A_{S \setminus \{i\}}}  v_S(a_i, \ba_{S \setminus \{i\}})$ for every $i \in S$ and $a_i' \in A_i$. This contest $\cC_S(a_i')$ awards its prize to any player who produces output along activity $a_i'$. The total prize that player $i$ gets from Type-1 and Type-5 contests is
\begin{align*}
    &\sum_{\ba_S \in A_S} v_S(\ba_S) -\sum_{2 \le \ell \le 5} \sum_{T \subseteq S, i \in T, |T| = \ell} \sum_{\ba_{S \setminus T} \in A_{S \setminus T}^-} \frac{\ell-1}{\ell} v_S(\ba_T^*, \ba_{S \setminus T}).
\end{align*}
As $\sum_{\ba_S \in A_S} v_S(\ba_S)$ does not depend upon the action $a_i^*$ selected by $i$, the utility of player $i$ is effectively
\begin{equation*}
    -\sum_{2 \le \ell \le 5} \sum_{T \subseteq S, i \in T, |T| = \ell} \sum_{\ba_{S \setminus T} \in A_{S \setminus T}^-} \frac{\ell-1}{\ell} v_S(\ba_T^*, \ba_{S \setminus T}).
\end{equation*}

\textbf{Type-2 Contests.} Let us add a contest named $\cC_S(a_i', a_j')$ with prize\\
$\sum_{\ba_{S \setminus \{i,j\}} \in A_{S \setminus \{i,j\}}} \frac{v_S(a_i',a_j',\ba_{S \setminus \{i,j\}})}{2}$ for every $i,j \in S, i \neq j$ and $a_i' \in A_i$ and $a_j' \in A_j$. This contest $\cC_S(a_i', a_j')$ awards its prize to any player who produces output along activity $a_i'$ or $a_j'$. The net utility of player $i$ gets from Type-5, Type-1, and Type-2 contests is
\begin{multline*}
    \sum_{3 \le \ell \le 5} \sum_{T \subseteq S, i \in T, |T| = \ell}  \sum_{\ba_{S \setminus T} \in A_{S \setminus T}^-}  \left(\frac{1}{2}\frac{4.\binom{3}{\ell-2}}{\binom{4}{\ell-1}} - \frac{\ell-1}{\ell} \right) v_S(\ba_T^*, \ba_{S \setminus T}) \\
    = \sum_{T \subseteq S, i \in T, |T| = 3} \sum_{\ba_{S \setminus T} \in A_{S \setminus T}^-}  \frac{1}{3} v_S(\ba_T^*, \ba_{S \setminus T}) \\
    + \sum_{T \subseteq S, i \in T, |T| = 4} \sum_{\ba_{S \setminus T} \in A_{S \setminus T}^-}  \frac{3}{4} v_S(\ba_T^*, \ba_{S \setminus T}) + \frac{6}{5} v_S(\ba_S^*).
\end{multline*}

\textbf{Type-3 Contests.} Let us add a contest named $\cC_S(\ba_T')$ with prize\\
$\sum_{\ba_T \neq \ba_T', \ba_{S \setminus T} \in A_{S \setminus T}} \frac{v_S(\ba_T, \ba_{S \setminus T})}{3}$ (effectively, the prize is $-\sum_{\ba_{S \setminus T} \in A_{S \setminus T}} \frac{v_S(\ba_T', \ba_{S \setminus T})}{3}$) for every $T \subset S$ with $|T|=3$ and for every $\ba_T' \in A_T$. This contest $\cC_S(\ba_T')$ awards its prize to any player who produces output along activities $\ba_T'$. The net utility of player $i$ gets from Type-5, Type-1, Type-2, and Type-3 contests is
\begin{multline*}
    \sum_{T \subseteq S, i \in T, |T| = 3} \sum_{\ba_{S \setminus T} \in A_{S \setminus T}^-}  \left(\frac{1}{3}-\frac{1}{3}\right) v_S(\ba_T^*, \ba_{S \setminus T}) \\
     + \sum_{T \subseteq S, i \in T, |T| = 4} \sum_{\ba_{S \setminus T} \in A_{S \setminus T}^-} \left(\frac{3}{4}-1\right) v_S(\ba_T^*, \ba_{S \setminus T}) + \left(\frac{6}{5}-2\right) v_S(\ba_S^*) \\
    = -\sum_{j \neq i} \sum_{a_j \in A_j^-} \frac{1}{4} v_S(\ba_{S \setminus \{j\}}^*, a_j) - \frac{4}{5} v_S(\ba_S^*).
\end{multline*}

\textbf{Type-4 Contests.} Let us add a contest named $\cC_S(\ba_T')$ with prize\\
$\sum_{\ba_{S \setminus T} \in A_{S \setminus T}} \frac{v_S(\ba_T', \ba_{S \setminus T})}{4}$ for every $T \subset S$ with $|T|=4$ and for every $\ba_T' \in A_T$. This contest $\cC_S(\ba_T')$ awards its prize to any player who produces output along activities $\ba_T'$. The net utility of player $i$ gets from all the five types of contests is
\[
    \frac{1}{5} v_S(\ba_S^*).
\]
We set $v_S(\ba_S) = 5 u_S(\ba_S)$ for every $\ba_S \in A_S$, and we are done.

\textbf{Cost.} 
It is easy to adapt the above proof to the single-activity cost model. All the players, contests, and activities are the same, we just need to replace the budgets with costs. In the proof above, we used the budgets to make sure that player $i$ can produce output only along activities in $A_i$ (note that we are in a single-activity model so they will produce output along exactly one activity). We can achieve that same for the cost model: for player $i$, set a very low cost for the activities in $A_i$ and very high cost for the activities in $A_j$, $j \neq i$, such that player $i$ always wants to produce output along one of the activities in $A_i$ and never along the activities not
in $A_i$.


\qed\end{proof}

\begin{proof}[Theorem~\ref{thm:plsCompleteness}]
We prove this result by reducing the problem of finding a locally optimal max-cut solution in a graph, which is known to be a PLS-complete problem~\cite{johnson1988easy}, to our problem. Let $(V,E)$ be an arbitrary graph and let $w_e \in \bZ_{\ge 0}$ be the weight of edge $e \in E$. A cut $(S,\overline{S})$ (where $\overline{S} = V \setminus S$) is locally optimal if one cannot improve the max-cut objective $\sum_{(u,v) \in E, u \in S, v \in \overline{S}} w_{u,v}$ by moving a vertex $u$ from its current partition $S$ (or $\overline{S}$) to the other partition $\overline{S}$ (or $S$). 

\textbf{Budget.} Let us first focus on the single-activity budget model, we shall slightly tweak the same construction for the single-activity cost model.

\textit{Construction.} We construct an instance of the budget game corresponding to the given graph $(V,E)$ as follows:
\begin{enumerate}
    \item Let there be $|V|$ players.
    \item Let there be $2|V|$ activities, which we identify by $\alpha_v$ and $\beta_v$ for $v \in V$.
    \item Let each player have a budget of $1$ that they can spend on any of the activities. As we are in a single-activity setup, in a pure strategy, a player would select an activity and produce an output of $1$ along that activity. Note that the strategy space is symmetric across the players.
    \item Let there be $|V| + 2|E|$ contests. 
    \begin{itemize}
        \item We identify the first $|V|$ contests using the vertices. The contest that corresponds to $v \in V$ equally distributes a large prize of $4\sum_{e \in E} w_e + 2$ to the players who produce a total output of at least $1$ along the two activities $\alpha_v$ and $\beta_v$ corresponding to $v$, i.e., a player $i$ receives a share of the prize if $b_{i,\alpha_v} + b_{i,\beta_v} \ge 1$. We call these contests \textit{vertex} contests.
        \item We identify the other $2|E|$ contests using the edges, two contests per edge. For an edge $e = (u,v) \in E$, one of the contests distributes a prize of $w_e$ to the players who produce an output of at least $1$ along the two activities $\alpha_u$ and $\alpha_v$, and the other contest distributes a prize of same value $w_e$ to the players who produce an output of at least $1$ along $\beta_u$ and $\beta_v$. We call these contests \textit{edge} contests.
    \end{itemize}
\end{enumerate}

\textit{Analysis.} Observe that the total prize awarded by the $2|E|$ edge contests is $2(\sum_{e \in E} w_e)$, while the prize awarded by a single vertex contest is $4\sum_{e \in E} w_e + 2$, which is much larger. As a player can produce output along only one activity, and as each vertex contest $v$ values a unique disjoint pair of activities $(\alpha_v, \beta_v)$, no player can win prizes from more than one vertex contest. Also, as the number of players is equal to the number of the vertex contests, each player can get a complete share of one of the vertex contests. In a PNE, each player selects a distinct pair $(\alpha_v, \beta_v)$ and produces output along one of the activities in the pair, because a player can move from a shared prize to an unshared prize increasing their utility by at least $2\sum_{e \in E} w_e + 1$ from the vertex contests, which dominates maximum total loss of $2\sum_{e \in E} w_e$ from the $2|E|$ edge contests. Let us denote the player who produces output along the pair of activities $(\alpha_v, \beta_v)$ as $i_v$.

Let $S = \{ v \in V \mid b_{i_v,\alpha_v} = 1 \}$ and $\overline{S} = V \setminus S = \{ v \in V \mid b_{i_v,\beta_v} = 1 \}$. For an edge $(u,v)$, if both players $i_u$ and $i_v$ are in $S$ or $\overline{S}$, then they share the prize of $w_{u,v}$, on the other hand if one of them is in $S$ and the other in $\overline{S}$ then they each get $w_{u,v}$. The value of the potential function (Equation~\eqref{eq:potFun}) $\phi$ can be written using $S$ and $\overline{S}$ as
\begin{multline*}
    \phi(S,\overline{S}) = \left(4\sum_{e \in E} w_e + 2\right) |V| + \left( \frac{3}{2} \sum_{e \in E} w_e + \frac{1}{2} \sum_{(u,v) \in E, u \in S, v \in \overline{S}} w_{u,v} \right) \\
    = \left(4 |V| + \frac{3}{2} \right) \sum_{e \in E} w_e + 2 |V| + \frac{1}{2}\sum_{(u,v) \in E, u \in S, v \in \overline{S}} w_{u,v},
\end{multline*}
where the contribution of $|V| (4\sum_{e \in E} w_e + 2)$ is from the $|V|$ vertex contests and a contribution of $3 w_{u,v}/2$ (if both $u$ and $v$ are in $S$ or $\overline{S}$) or $2 w_{u,v}$ (if one of $u$ or $v$ is in $S$ and the other in $\overline{S}$) is from the two contests corresponding to the edge $(u,v)$. As we have argued before, a player $i_v$ will never change their associated vertex $v$ in a better response move. A move by the player $i_v$ of shifting from $\alpha_v$ to $\beta_v$, or vice-versa, corresponds to moving the vertex $v$ across the cut $(S, \overline{S})$ in the max-cut problem. The value of the potential is effectively equal to the objective of the max-cut problem: we have multiplied the max-cut objective with a positive constant factor of $1/2$ and added a constant term. So, we have a one-one correspondence between local search in max-cut and dynamics in the single-activity budget game.

\textbf{Cost.} Note that for the budget game constructed above, the budget constraint was not very crucial in the analysis. No player had an incentive to produce an output of strictly more than $1$ along the activity they selected, because all the contests already awarded the prizes as soon as the players produced an output of $1$. In the cost model, we remove the budget constraint and set $c_{i, \ell} = 4\sum_{e \in E} w_e + 2$ for every player $i$ and activity $\ell$. The same analysis used in the budget model holds here too, the value of the potential function (Equation~\eqref{eq:potFun}) here is:
\begin{multline*}
    \phi(S,\overline{S}) = \left(4\sum_{e \in E} w_e + 2\right) |V| \\
    + \left( \frac{3}{2} \sum_{e \in E} w_e + \frac{1}{2} \sum_{(u,v) \in E, u \in S, v \in \overline{S}} w_{u,v} \right) - \left(4\sum_{e \in E} w_e + 2\right) |V| \\
    = \frac{3}{2} \sum_{e \in E} w_e + \frac{1}{2} \sum_{(u,v) \in E, u \in S, v \in \overline{S}} w_{u,v},
\end{multline*}
where the positive contribution of $|V| (4\sum_{e \in E} w_e + 2)$ is from the $|V|$ vertex contests, the contribution of $3 w_{u,v}/2$ or $2 w_{u,v}$ is from the two contests corresponding to the edge $(u,v)$, and a negative contribution of $|V| (4\sum_{e \in E} w_e + 2)$ is due to the cost of the players. The rest of the argument follows as before.

\textbf{Mixed-strategy Nash equilibrium.} In the {\sc Max Cut} problem, a randomized solution where each vertex is on either side of the cut with a probability of $1/2$ cannot be improved by moving just one vertex to a different side (or to a different distribution over the two sides). So, it is locally stable. We apply the same idea to find an MNE in the single-activity contest game just constructed.

Consider the following mixed strategy profile: for every $v \in V$, there is a unique player $i_v$ who produces an output of $1$ along one of the two activities $\alpha_v$ or $\beta_v$ with a probability of $1/2$ each. As argued before, player $i_v$ would not want to produce output along any other $\alpha_{u}$ or $\beta_u$ for $u \neq v$ because of the vertex contests. Also, given that every other player is following this strategy, player $i_v$ has the same utility for selecting either $\alpha_v$ or $\beta_v$, so $(1/2, 1/2)$ is as good a mixed-strategy as any other. So, we have an equilibrium.
\qed\end{proof}

\subsubsection{Fixed Parameter Tractability for Single-Activity Models}
\begin{theorem}\label{thm:fptSingleActivityN}
In the single-activity models, both cost and budget, we can compute a pure-strategy Nash equilibrium in polynomial time if the number of players is a constant. 
\end{theorem}
\begin{proof}[Theorem~\ref{thm:fptSingleActivityN}]
For the budget model, a player essentially chooses an activity among the $k$ activities; once she selects an activity, she can w.l.o.g. produce the maximum output in her budget for that activity. So, she essentially has $k$ actions.

Similarly, for the cost model, she essentially has at most $m \cdot k$ actions. First, the player chooses an activity among the $k$ activities, then she chooses the level of output for that activity. For a given activity, the criterion of a contest corresponds to a minimum level of output. When choosing the level of output for an activity, a player can w.l.o.g. choose among these minimum levels of output corresponding to the contests. As she increases her output for an activity from one level to another level with a higher output, she wins a superset of contests as she was winning previously.

As there are $n$ players and each player (effectively) has at most $m \cdot k$ actions ($k \le m \cdot k$ for budget model), therefore there are at most $(m \cdot k)^n$ action profiles. We can enumerate over these profiles and find the PNE is polynomial time if $n$ is a constant.
\qed\end{proof}

\begin{theorem}\label{thm:fptSingleActivityM}
In the single-activity budget model, we can compute a pure-strategy Nash equilibrium in polynomial time if the number of contests is a constant. 
\end{theorem}
\begin{proof}[Theorem~\ref{thm:fptSingleActivityM}]
If we carefully observe the potential function of the budget model, $\sum_{j \in M} \sum_{\ell \in [n_j(\bb)]} v_j(\ell)$, given in \eqref{eq:potFun}, we can see that for every contest $j \in M$, $\sum_{\ell \in [n_j(\bb)]} v_j(\ell)$ can take at most $n+1$ different values depending upon the value of $n_j(\bb)$. So, the potential function can take at most $(n+1)^m$ different values. 

If we are not at a PNE, there is a best-response move that strictly increases the potential function, and we can find such a best-response in polynomial time (even for the multi-activity model, see Theorem~\ref{thm:fptBestResponse}). So, we can find a PNE in polynomial time by repeatedly finding a best response move and strictly increasing the potential function, and this process will terminate in at most $(n+1)^m$ many steps.
\qed\end{proof}
Note that idea used in the proof of Theorem~\ref{thm:fptSingleActivityM} does not directly carry over to the cost model because of the additional $\left( - \sum_{i \in N} c_i(\bb_i) \right)$ term in the potential function for the cost model.

\end{document}